\documentclass[journal]{IEEEtran}

\ifCLASSINFOpdf
\else
   \usepackage[dvips]{graphicx}
\fi
\usepackage{url}
\usepackage{stfloats}
\usepackage{graphicx}
\usepackage{cite}
\usepackage{amsfonts}
\usepackage{amssymb}
\usepackage{amsthm}
\usepackage{amsmath}
\usepackage{xcolor}

\newcommand{\ben}{\begin{enumerate}}
\newcommand{\een}{\end{enumerate}}
\newcommand{\beq}{\begin{equation}}
\newcommand{\eeq}{\end{equation}}
\newcommand{\bde}{\begin{description}}
\newcommand{\ede}{\end{description}}

\newcommand{\argmin}{\operatornamewithlimits{argmin}}

\newcommand{\norm}[1]{\lVert#1\rVert}

\newtheoremstyle{slplain}
  {1\baselineskip\@plus.2\baselineskip\@minus.2\baselineskip}
  {.5\baselineskip\@plus.2\baselineskip\@minus.2\baselineskip}
  {\slshape}
  {}
  {\bfseries}
  {.}
  { }
  {}

\theoremstyle{slplain}

\newtheorem{theorem}{Theorem}

\newtheorem{definition}[theorem]{Definition}

\newtheorem{proposition}[theorem]{Proposition}

\usepackage{booktabs,array}

\newcount\rowc

\begin{document}

\title{A Laplace Mixture Representation of the Horseshoe and Some Implications}

\author{Ksheera Sagar and Anindya Bhadra
\thanks{Submitted on 11/28/2022. Bhadra is supported by US National Science Foundation Grant DMS-2014371.}
\thanks{The authors are with the Department of Statistics, Purdue University, IN 47906, United States of America  (e-mails: kkeralap@purdue.edu; bhadra@purdue.edu).}}
\maketitle

\begin{abstract}
The horseshoe prior, defined as a half Cauchy scale mixture of normal, provides a state of the art approach to Bayesian sparse signal recovery. We provide a new representation of the horseshoe density as a scale mixture of the Laplace density, explicitly identifying the mixing measure. Using the celebrated Bernstein--Widder theorem and a result due to Bochner, our representation immediately establishes the complete monotonicity of the horseshoe density and strong concavity of the corresponding penalty. Consequently, the equivalence between local linear approximation and expectation--maximization algorithms for finding the posterior mode under the horseshoe penalized regression is established. Further, the resultant estimate is shown to be sparse.
\end{abstract}


\IEEEpeerreviewmaketitle

\section{Introduction}
\IEEEPARstart{F}{or} response $y$, covariates $\Phi$ and regression coefficients $x$, consider the high-dimensional linear regression model:
\begin{eqnarray}
y=\Phi x+ \epsilon,\label{regression_model}
\end{eqnarray}
where $y=\{y_j\}\in \mathbb{R}^n, \Phi=\{\phi_{ji}\}\in \mathbb{R}^{n\times p}, x=\{x_i\}\in\mathbb{R}^{p}$ with $p\gg n$ and $\epsilon \sim \mathcal{N}(0,\sigma^2\mathbf{I}_n)$.
The horseshoe prior hierarchy \cite{carvalho2010horseshoe} is defined for $i=i,\ldots,p$ as a half Cauchy scale mixture of normal as:
\begin{equation}
x_i \mid \lambda_i, \tau \stackrel{ind}\sim \mathcal{N}(0,\lambda_i^2\tau^2),\; \lambda_i\stackrel{ind}\sim\mathcal{C}^{+}(0,1), \; \tau>0, \label{eq:prior}
\end{equation}
where $\mathcal{C}^{+}$ denotes a standard half Cauchy random variable with density $p(\lambda_i) = (2/\pi)(1+ \lambda_i^2)^{-1},\; \lambda_i>0$ and $\tau$ is a hyperparameter. This prior has found numerous applications in sparse signal recovery\cite{carvalho2009handling, carvalho2010horseshoe, li2019graphical, yu2019variational, bhadra2019lasso, bhadra2020horseshoe, lee2020variation, bhadra2021horseshoe, li2021joint}. Although there is no closed form to the marginal horseshoe density, obtained by integrating out the latent $\lambda_i$s, the main benefits of the horseshoe stem from two crucial properties of the marginal prior. The first is infinite prior mass at the origin, allowing strong shrinkage towards zero in a sparse regime; and the second is heavy, polynomially decaying tails, aiding detection of signals \cite{carvalho2010horseshoe}. 

Since the tails of the horseshoe prior follow a quadratic rate of decay, it is conceivable to expect the prior density to admit a mixture representation with respect to the double exponential (Laplace) density as well. Yet, to our knowledge, no such explicit representation is available in the literature. We close this gap by showing the horseshoe density can also be represented as a mixture of double exponential with respect to a random variable whose density is proportional to the Dawson function. This representation, in turn, allows us to appeal to the celebrated Bernstein--Widder theorem \cite{widder2015laplace} and a result due to Bochner \cite{bochner2005harmonic} to establish complete monotonicity of the horseshoe density and  strong concavity of the horseshoe penalty, both new results. Further, appealing to the same mixture representation, we are able to establish an equivalence between local linear approximation or LLA \cite{zou2008one,polson2016mixtures} and expectation--maximization or EM \cite{dempster1977maximum} schemes for finding the maximum a posteriori (MAP) estimate under the horseshoe prior for sparse signal recovery.

\section{A Laplace Mixture Representation of the Horseshoe}
\begin{proposition}
\label{prop_equiv_HS_Laplace}
The marginal horseshoe density for a scalar random variable ($p=1$) admits the following representation as a Laplace mixture:
\begin{eqnarray*}
p_{HS}(x) &=& \frac{{2}}{\pi\sqrt{\pi}\tau} \int_0^\infty \exp\left(-u\frac{|x|}{\tau}\right) D_{+}\left(\frac{u}{\sqrt{2}}\right) du,
\end{eqnarray*}
where $D_{+}(z)=\exp(-z^2)\int_0^z\exp(t^2)dt,\; z>0,$ is the Dawson function.
\end{proposition}

\begin{proof}
We have using the hierarchy of Equation~\eqref{eq:prior},
\begin{equation*}
\begin{split}
p_{HS}(x)=& \int_0^\infty \frac{1}{\sqrt{2\pi\lambda^2\tau^2}}  \exp\left(-\frac{x^2}{2\lambda^2 \tau^2}\right) \frac{2}{\pi} \frac{1}{1+\lambda^2} d\lambda
\end{split}
\end{equation*}
\vspace{-0.2cm}
\begin{equation*}
\begin{split}
=& \frac{\sqrt{2}}{\pi\sqrt{\pi}\tau} \int_0^\infty  \frac{1}{\lambda} \exp\left(-\frac{x^2}{2\lambda^2\tau^2}\right)  \int_0^\infty \exp(-(\lambda^2 +1) t) dt d\lambda\\
=& \frac{\sqrt{2}}{\pi\sqrt{\pi}\tau} \int_0^\infty \int_0^\infty  \frac{1}{\lambda} \exp\left(-\frac{x^2}{2\lambda^2\tau^2} - \lambda^2 t \right)  \exp(-t) dt d\lambda\\
=& \frac{\sqrt{2}}{\pi\sqrt{\pi}\tau} \int_0^\infty \left\{\int_0^\infty  \exp\left(-\frac{x^2}{2\omega\tau^2} - \omega t \right)\frac{1}{2} \omega^{-1} d\omega\right\} \exp(-t) dt\\
& (\text{substituting }\lambda^2 = \omega, \text{and using Fubini's theorem})\\
=& \frac{1}{\pi\sqrt{2\pi}\tau} \int_0^{\infty} 2 K_0\left(\sqrt{2t} \frac{|x|}{\tau}\right) \exp(-t) dt\\
&\text{(generalized inverse Gaussian integral \cite{barndorff1977exponentially};}\\
&\text{$K_0=$ modified Bessel function of the second kind of order 0)}\\
=&  \frac{\sqrt{2}}{\pi\sqrt{\pi}\tau} \int_0^{\infty} \left\{\int_0^{\infty} \exp\left(-\frac{|x|}{\tau} \sqrt{2t} \cosh \zeta\right) d\zeta\right\} \exp(-t) dt\\
&\text{(using identity 9.6.24 in \cite{abramowitz1964handbook})}\\
=&  \frac{\sqrt{2}}{\pi\sqrt{\pi}\tau} \int_0^{\infty} \int_0^{\infty} \exp\left(-\frac{|x|}{\tau} \sqrt{2t} \cosh \zeta -t \right) d\zeta dt.
\end{split} 
\end{equation*}

\noindent Let $\sqrt{2t} \cosh \zeta = u$ and $t=v$. The inverse transformations are, $\cosh^{-1} (u/\sqrt{2v}) = \zeta$ and $v=t$; with Jacobian:
\begin{eqnarray*}
J&=& \begin{vmatrix}
    \frac{\partial \zeta} {\partial u}  & \frac{\partial \zeta} {\partial v} \\
    & \\
    \frac{\partial t} {\partial u}& \frac{\partial t} {\partial v} 
  \end{vmatrix}= \begin{vmatrix}
    \frac{\partial \zeta} {\partial u}  & \frac{\partial \zeta} {\partial v} \\
    & \\
    0 & 1
  \end{vmatrix}= \left|\frac{\partial\zeta} {\partial u}\right| =  \frac{1}{\sqrt{u^2 - 2v}}.
\end{eqnarray*}
Thus, $p_{HS}(x)$ equals,
\begin{eqnarray*}
&&\frac{\sqrt{2}}{\pi\sqrt{\pi}\tau} \int_0^{\infty} \int_0^{u^2/2} \exp\left(-\frac{|x|}{\tau}u -v \right) \frac{1}{\sqrt{u^2 - 2v}} dv du\\
&=& \frac{\sqrt{2}}{\pi\sqrt{\pi}\tau} \int_0^{\infty} \exp\left(-\frac{|x|}{\tau}u\right)   \left\{\int_0^{u^2/2}  \frac{\exp(-v )}{\sqrt{u^2 - 2v}} dv \right\} du\\
&=& \frac{{2}}{\pi\sqrt{\pi}\tau} \int_0^{\infty} \exp\left(-\frac{|x|}{\tau}u\right) \left\{  e^{-u^2/2}\int_0^{\frac{u}{\sqrt{2}}}e^{t^2}dt \right\} du\\
&&(\text{substituting }u^2-2v=2t^2)\\
&=& \frac{{2}}{\pi\sqrt{\pi}\tau} \int_0^{\infty} \exp\left(-\frac{|x|}{\tau}u\right) D_{+}\left(\frac{u}{\sqrt{2}}\right) du. \qedhere\\
\end{eqnarray*}
\end{proof}
Proposition~\ref{prop_equiv_HS_Laplace} is our main result and the remainder of the paper discusses its implications in sparse signal recovery. Numerical validation of Proposition~\ref{prop_equiv_HS_Laplace} is presented in Fig.~\ref{fig:validation} by evaluating the horseshoe density using both the usual normal scale mixture (cf. Equation~\eqref{eq:prior}) and the newly derived Laplace scale mixture representations over $[-2,2]$. The point-wise maximum difference in density evaluation (in logarithmic scale) between the two representations is less than $10^{-10}$.   
\begin{figure*}[tb]
\centering
\includegraphics[width=\textwidth, height =0.18\textheight]{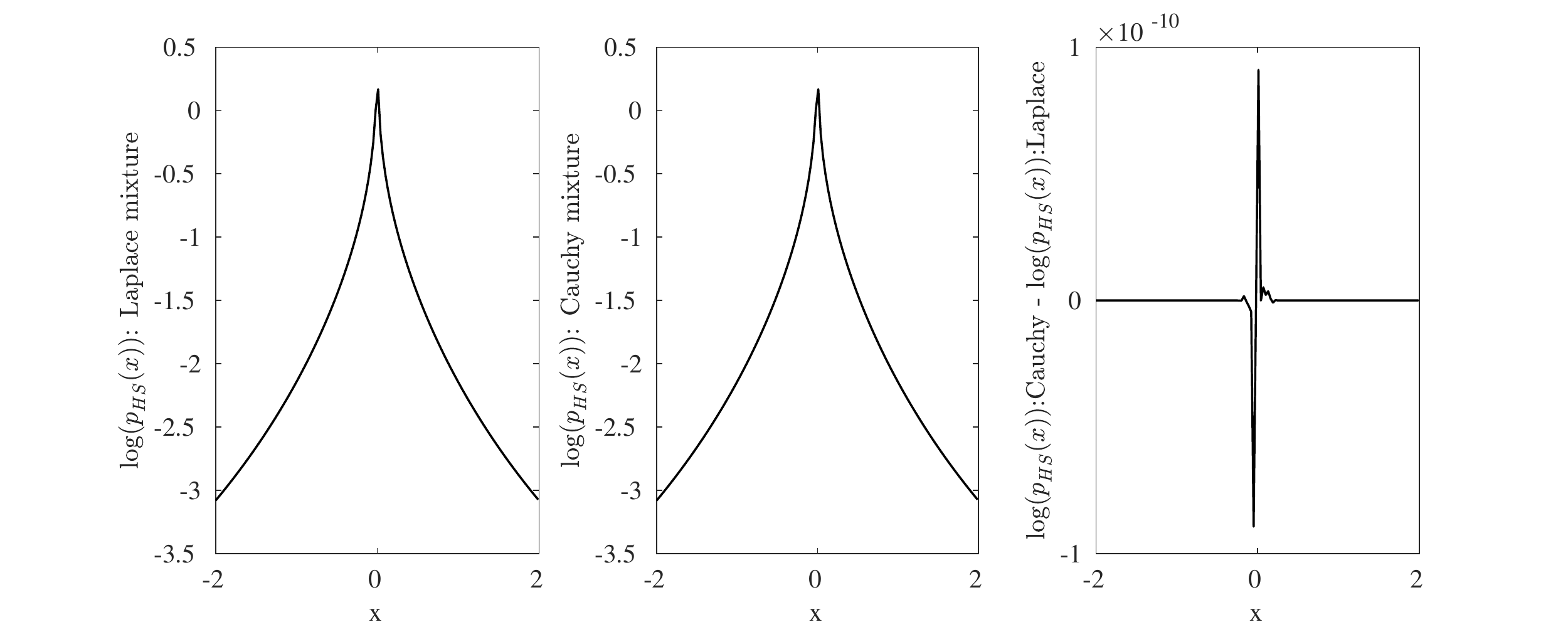}
\caption{
\label{fig:validation}
Numerical validation of Proposition~\ref{prop_equiv_HS_Laplace}. Left panel: logarithm of horseshoe density, $\log p_{HS}(x)$, evaluated numerically, using the Laplace scale mixture in Proposition~\eqref{prop_equiv_HS_Laplace}. Middle panel: logarithm of horseshoe density, $\log p_{HS}(x)$, evaluated numerically, using the normal scale mixture in Equation~\eqref{eq:prior}. Right panel: Difference between the logarithms of densities evaluated using normal and Laplace mixtures.} 
\end{figure*}

\section{Main Implications of Proposition~\ref{prop_equiv_HS_Laplace}: Complete Monotonicity of 
\label{main_implications}
the Horseshoe Density and Strong Concavity of the Penalty}
We begin by stating the definition of a completely monotone function and the celebrated Bernstein--Widder theorem \cite[Chapter 4]{widder2015laplace}.
\begin{definition}
A function $g(\cdot)$ is said to be completely monotone if $(-1)^k g^{(k)}(x)\ge 0$ for $x\ge0,\; k=0,1,\ldots,$ with superscripts $(k)$ denoting the order of derivatives and $g^{(0)}\equiv g$.
\end{definition}
\begin{theorem}\label{th:bernstein} (Bernstein--Widder). The function $g(\cdot)$ is completely monotone if and only if it is the Laplace transform of a positive measure $\gamma$, i.e., admits the representation: $g(x) = \int_{0}^{\infty} \exp(-xu) d\gamma(u),\; x\ge 0$.
\end{theorem}
Since the Dawson function is positive over $(0,\infty)$, and the horseshoe density is symmetric around zero, Proposition~\ref{prop_equiv_HS_Laplace} in conjunction with Theorem~\ref{th:bernstein} immediately establishes the complete monotonicity of $p_{HS}(|x|)$.

In a penalized likelihood setting, it is customary to view the negative log likelihood as the loss function and the negative log prior density as the penalty in an optimization problem. The following result, due to Bochner \cite[Theorem 4.1.5]{bochner2005harmonic} establishes the strong concavity of the horseshoe penalty.
\begin{theorem}\label{lem:bochner} (Bochner). The function $g(x) = \exp(-a \phi(x)), x\ge 0,$ is completely monotone for every $a>0$ if and only if $\phi'(x)=\phi^{(1)}(x)$ is completely monotone.
\end{theorem}
Since $p_{HS}(|x|)$ is completely monotone, Theorem~\ref{lem:bochner} implies the corresponding penalty function $\mathrm{pen}_{HS}(|x|)=-\log p_{HS}(|x|)$  is a Bernstein function (integral of a completely monotone function) and thus strongly concave. In other words, $\mathrm{pen'}_{HS}(|x|) = -(p'_{HS}(|x|))/(p_{HS}(|x|))$ is completely monotone. This property of the penalty has important implications on the convergence of optimization algorithms for penalized likelihood approaches, which we discuss next.

\section{Sparse Signal Recovery via LLA}
\subsection{Local Linear Approximation of Non-convex Penalties}
Consider the penalized estimation problem:
\begin{equation}
\argmin_{x}\left(\ell(x; y) + \mathrm{pen}(|x|)\right), \label{eq:pen}
\end{equation}
where $\ell(x;y)= -\log \mathcal{L}(x;y)$ is a loss function providing some measure of model fit to the data $y$, usually a negative log likelihood function and $\mathrm{pen}(|x|)$ is a penalty symmetric around zero, with $|\cdot|$ denoting the vector $\ell_1$ norm. Although convex penalties are simpler to study, both theoretically and computationally, concave penalties enjoy several desirable statistical properties such as near-unbiasedness, minimax optimal properties and sparsity~\cite{fan2001variable, fan2014strong}. Nevertheless, the optimization problem becomes much more challenging, and usually, it is only possible to find a local optimum. Required conditions for global optimum under concave penalties are discussed by~\cite{mazumder2011sparsenet}.

The local linear approximation (LLA) algorithm of \cite{zou2008one} provides one approach for solving the problem in Equation~\eqref{eq:pen} for concave penalties. LLA iterations proceed via a first order Taylor approximation of the penalty function, as follows
\begin{equation*}
   \mathrm{pen}(|x|) \approx   \mathrm{pen}(|x_{(k)}|) + \mathrm{pen'}(|x_{(k)}|)(|x|-|x_{(k)}|),
\end{equation*}
\noindent where $x_{(k)}$ is the current parameter value. With this approximation, and assuming the loss (but not necessarily the penalty) is convex, the non-convex optimization problem in Equation~\eqref{eq:pen} reduces to the convex problem
\begin{equation}
\label{skeleton_beta_update}
    x_{(k+1)} = \underset{x}{\text{argmin}}\left(\ell(x;y)+\lambda_{(k)}|x|\right),
\end{equation}
where $\lambda_{(k)}=\mathrm{pen'}(|x_{(k)}|)$. With the common squared error loss function, $\ell(x;y) = |x-y|^2$, Equation~\eqref{skeleton_beta_update} defines a lasso problem \cite{tibshirani1996regression}, which can be solved very efficiently. In addition, the resultant solution is naturally sparse.

\subsection{Implications of Complete Monotonicity on LLA}
Numerical convergence and statistical properties of LLA estimates are non-trivial objects of study for general non-convex penalties and properties of LLA for general penalties are largely unknown. Nevertheless, in the special case when the penalty under consideration is a Bernstein function, certain simplifications ensue due to the following result of \cite{zou2008one}.
\begin{theorem} (Zou and Li). 
Consider a penalty $\mathrm{pen}(|x|)$ such that $\exp(-\mathrm{pen}(|x|))$ is completely monotone. Then, under a Gaussian likelihood, the LLA algorithm is equivalent to an EM algorithm for MAP estimation under a prior, not necessarily proper, with density $p_{X}(x)\propto \exp(-\mathrm{pen}(x))$.
\end{theorem}
Thus, an application of Proposition~\ref{prop_equiv_HS_Laplace} immediately establishes the LLA algorithm as an EM algorithm for horseshoe penalties. This result is of practical significance, since the convergence properties of LLA can then be studied using the same tools for EM, which are better understood~\cite{wu1983convergence}. This equivalence is also established by \cite{polson2016mixtures}, who provide a variational interpretation. Thus, it is established that the MAP estimate under the horseshoe prior is sparse; due to the soft thresholding operator in LLA resulting in exact zeros. Further, to our knowledge, this is the first viable algorithm for finding the MAP estimate under the horseshoe.

\subsection{Implementing LLA under the Horseshoe Penalty}
Consider now the practicalities of LLA under the horseshoe penalty. For implementation, one needs to compute $\lambda_{(k)}=\mathrm{pen'}_{HS}(|x_{(k)}|)$. Differentiating the expression of Proposition~\ref{prop_equiv_HS_Laplace} under the integral sign, we have:
\begin{equation}
\label{derivative_peanlty}
   \lambda_{(k)} = \mathrm{pen'}_{HS}(|x_{(k)}|) = \frac{\int_{0}^\infty \frac{u}{\tau} \exp\left(-u\frac{|x_{(k)}|}{\tau}\right)D_+\left(\frac{u}{\sqrt{2}}\right)du}{\int_{0}^\infty\exp\left(-u\frac{|x_{(k)}|}{\tau}\right)D_+\left(\frac{u}{\sqrt{2}}\right)du}.
\end{equation}
We provide numerical details of evaluating $\lambda_{(k)}$ in the Appendix. Subsequent to this, we use efficient solvers based on pathwise coordinate optimization~\cite{friedman2007pathwise}, for solving the lasso problem in Equation~\eqref{skeleton_beta_update}. We note here that using the expression from Equation~\eqref{derivative_peanlty} is not necessary for evaluating $\mathrm{pen'}_{HS}$, strictly speaking, and other alternative numerical methods can be used. However, in our experience, Equation~\eqref{derivative_peanlty} is numerically very stable.

\section{Numerical Experiments}\label{sec:num}
\subsection{Sparse Normal Means Model}
\label{normal_means_model}

Revisiting the regression model in Equation~\eqref{regression_model}, and setting $p=n,\,\Phi =\mathbf{I}_n$, we recover the normal means model $y_i\stackrel{ind}\sim\mathcal{N}(x_i,\sigma^2)$. In our experiments, we set $n=50,\,\sigma^2=1$ and the components in $x$ as, $x_1=x_2=3,\,x_3,\ldots,x_{50}=0$; $y_i\text{ is sampled from }\mathcal{N}(x_i,\sigma^2)$. With the generated data, optimizing Equation~\eqref{skeleton_beta_update} with $\ell(x;y) = \frac{1}{2}\sum_{i=1}^n(y_i - x_i)^2$, with pathwise coordinate optimization~\cite{friedman2007pathwise}, we have :
\begin{equation*}
    \begin{split}
         x_{(k+1)} &= \underset{x}{\text{argmin}}\left(\frac{1}{2}\sum_{i=1}^n(y_i-x_i)^2+\sum_{i=1}^n\lambda_{i(k)}|x_i|\right),\\
    \text{yielding, }  x_{i(k+1)} &= S(y_i,\lambda_{i(k)}),\text{ where},
    \end{split}
\end{equation*}
\begin{equation}
    \label{soft_threshold}
    S(\beta,\gamma) = \begin{cases}\beta - \gamma,&\text{ if }\beta>0\text{ and } \gamma<|\beta|,\\
    \beta + \gamma,&\text{ if }\beta<0\text{ and } \gamma<|\beta|,\\
    0, &\text{ if }\gamma\geq|\beta|,
    \end{cases}
\end{equation}
is the soft thresholding operator. In the above, $x_{i(k+1)}$ denotes the value of $i^{th}$ component of $x$, at $(k+1)^{th}$ iteration. Similarly, $\lambda_{i(k)}$ denotes the derivative of the penalty in Equation~\eqref{derivative_peanlty}, evaluated at $x_{i(k)}$. We estimate the normal means, $\hat{x}=\{\hat{x}_i\}$, under the horseshoe penalty using LLA; and compare with the posterior mean of  $x$ under the horseshoe prior obtained by Markov chain Monte Carlo (HS-MCMC), the minimax concave penalty or MCP \cite{zhang2010nearly}, the smoothly clipped absolute deviation penalty or SCAD \cite{fan2001variable} and the lasso \cite{tibshirani1996regression} penalties, using the \texttt{R} packages \texttt{horseshoe}~\cite{van2016horseshoe}, \texttt{sparsenet}~\cite{mazumder2011sparsenet}, \texttt{ncvreg}~\cite{ncvreg} and \texttt{glmnet}~\cite{glmnet} respectively. MCP are SCAD are non-convex penalties, while lasso is convex. Under all the point estimation methods, the respective tuning parameters are chosen using 10-fold cross validation; whereas for HS-MCMC, we consider $\tau\sim\mathcal{C}^+(0,1)$, a total of $1.5\times 10^4$ posterior samples with 5000 samples as burnin, and 50\% credible intervals for variable selection~\cite{li2019graphical}. For the LLA procedure, we start with $x_{i(0)}=1,\,\forall  i$ and set the stopping criterion as ${\sum_{i=1}^n\left(\hat{x}_{i(k+1)}-\hat{x}_{i(k)}\right)^2}<10^{-6}$. To compare the estimates, we compute the in sample sum of squared errors or $\text{SSE}=\sum_{i=1}^n(\hat{x}_i-x_i)^2$, out of sample prediction SSE (pSSE) $=\norm{y_\mathrm{out}-\Phi\hat{x}}^2$; for $y_\mathrm{out}$ generated similarly as $y$, true negative (positive) rates, for zeros (non-zeros) in $x$ identified correctly as TNR, TPR  respectively and the computational time in seconds. We estimate $\hat{x}$ for 50 replications, and present the mean and standard deviation (in parentheses) of the performance measures in Table~\ref{normal_means_result}. We observe that HS-LLA has the best pSSE and its SSE and TNR are competitive with the resultant MCMC estimate. The good statistical performance of the posterior mean estimate obtained by MCMC is not surprising, since it is the optimal Bayes estimate under $\ell_2$ loss. However, the MCMC estimate is about 3 times slower than than HS-LLA, for negligible gain in statistical performance.  
\begin{table}[!htb]
\centering
\caption{Comparison of results for competing procedures, in a sparse normal means model.}
\label{normal_means_result}
\begin{tabular}{|c|ccccc|}
\hline
         & HS-LLA           & HS-MCMC                                     & SCAD                                                      & MCP                                                       & lasso                                                     \\ \hline
SSE      & \begin{tabular}[c]{@{}c@{}}10.051\\ (6.721)\end{tabular}   & \begin{tabular}[c]{@{}c@{}}9.757\\ (3.803)\end{tabular} & \begin{tabular}[c]{@{}c@{}}26.22\\ (18.629)\end{tabular}  & \begin{tabular}[c]{@{}c@{}}32.724\\ (17.794)\end{tabular} & \begin{tabular}[c]{@{}c@{}}28.121\\ (18.177)\end{tabular}  \\
pSSE     & \begin{tabular}[c]{@{}c@{}}59.238\\ (16.771)\end{tabular}  & \begin{tabular}[c]{@{}c@{}}61.381\\ (12.506)\end{tabular}& \begin{tabular}[c]{@{}c@{}}75.641\\ (25.609)\end{tabular} & \begin{tabular}[c]{@{}c@{}}80.358\\ (22.146)\end{tabular} & \begin{tabular}[c]{@{}c@{}}76.082\\ (23.228)\end{tabular} \\
TNR      & \begin{tabular}[c]{@{}c@{}}0.982\\ (0.023)\end{tabular}    & \begin{tabular}[c]{@{}c@{}}0.992\\ (0.013)\end{tabular}& \begin{tabular}[c]{@{}c@{}}0.633\\ (0.336)\end{tabular}   & \begin{tabular}[c]{@{}c@{}}0.65\\ (0.325)\end{tabular}    & \begin{tabular}[c]{@{}c@{}}0.464\\ (0.366)\end{tabular}   \\
TPR      & \begin{tabular}[c]{@{}c@{}}0.64\\ (0.351)\end{tabular}     & \begin{tabular}[c]{@{}c@{}}0.53\\ (0.409)\end{tabular}& \begin{tabular}[c]{@{}c@{}}0.86\\ (0.268)\end{tabular}    & \begin{tabular}[c]{@{}c@{}}0.84\\ (0.31)\end{tabular}     & \begin{tabular}[c]{@{}c@{}}0.93\\ (0.202)\end{tabular}     \\
Time (s) & 0.351  & 1.07 & 0.268 & 0.31 & 0.202\\ \hline
\end{tabular}
\end{table}

\subsection{Sparse Linear Regression}
\label{linear_regression}
 We set $n=50,\,p=100$ and the components $x$ as $x_1 =3,\,x_2 = 1.5,\,x_3=2,\,x_4,\ldots,x_{100}=0$. The entries in columns of $\Phi$, where $\Phi=[\Phi_1,\ldots,\Phi_p]$ are generated from a multivariate Gaussian, with zero mean, such that the covariance between $\Phi_k$ and $\Phi_l$, for $1\leq k,\,l\leq p$ is $0.5^{|k-l|}$.  The observations $y_j$ are generated as $y_j\sim\mathcal{N}(\Phi x,1)$, for $j\in\{1,\ldots,n\}$. With the generated data, optimizing~\eqref{skeleton_beta_update} with $\ell(x;y) = \frac{1}{2}\sum_{i=1}^n(y_i -\Phi x)^2$, with pathwise coordinate optimization~\cite{friedman2007pathwise}, we have :
\begin{equation*}
    \begin{split}
         x_{(k+1)} =& \underset{x}{\text{argmin}}\left(\frac{1}{2}\sum_{j=1}^n(y_i-\Phi x)^2+\sum_{i=1}^p\lambda_{i(k)}|x_i|\right),\\
    \text{yielding, }  x_{i(k+1)} &= \frac{1}{\sum_{j=1}^n \phi_{ji}^2}S\left(\sum_{j=1}^n \phi_{ji}\left(y_j - \widetilde{y}_{ji}\right),\lambda_{i(k)}\right),\\
   \text{where},\, \widetilde{y}_{ji}& = \underset{l\neq i}{\sum}\phi_{jl}\widetilde{x}_{l}\,,\,\widetilde{x}_{l}=\begin{cases}x_{l(k+1)},&\text{ if }i>l,\\ x_{l(k)},& \text{ else.}\end{cases}
    \end{split}
\end{equation*}
and the definition of $S(\cdot,\,\cdot)$ follows from Equation~\eqref{soft_threshold}. For the LLA procedure, we start with $x_{i(0)}=0.1,\,\forall  i$ and set the stopping criterion as mentioned in Section~\ref{normal_means_model}. We estimate $\hat{x}$ for 50 replications, and present the mean and standard deviation (in parentheses) of the performance measures in Table~\ref{linear_regression_results}. From these results, we observe that HS-MCMC has the best TNR, followed by HS-LLA and other measures under HS-LLA are comparable to the best performer, which is SCAD for this setting. 
\begin{table}[!b]
\centering
\caption{Comparison of results for competing procedures, in sparse linear regression.}
\label{linear_regression_results}
\begin{tabular}{|c|ccccc|}
\hline
         & HS-LLA    & HS-MCMC                                                & SCAD                                                      & MCP                                                       & lasso                                                     \\ \hline
SSE      & \begin{tabular}[c]{@{}c@{}}0.207\\ (0.2)\end{tabular}      & \begin{tabular}[c]{@{}c@{}}0.183\\ (0.175)\end{tabular}& \begin{tabular}[c]{@{}c@{}}0.118\\ (0.108)\end{tabular}   & \begin{tabular}[c]{@{}c@{}}0.182\\ (0.166)\end{tabular}   & \begin{tabular}[c]{@{}c@{}}0.281\\ (0.165)\end{tabular}   \\
pSSE     & \begin{tabular}[c]{@{}c@{}}59.612\\ (12.918)\end{tabular}  & \begin{tabular}[c]{@{}c@{}}58.038\\ (11.45)\end{tabular}& \begin{tabular}[c]{@{}c@{}}56.607\\ (11.753)\end{tabular} & \begin{tabular}[c]{@{}c@{}}59.233\\ (14.665)\end{tabular} & \begin{tabular}[c]{@{}c@{}}65.817\\ (16.688)\end{tabular} \\
TNR      & \begin{tabular}[c]{@{}c@{}}0.987\\ (0.019)\end{tabular}    & \begin{tabular}[c]{@{}c@{}}1\\ (0)\end{tabular}& \begin{tabular}[c]{@{}c@{}}0.975\\ (0.026)\end{tabular}   & \begin{tabular}[c]{@{}c@{}}0.975\\ (0.064)\end{tabular}   & \begin{tabular}[c]{@{}c@{}}0.92\\ (0.077)\end{tabular}    \\
TPR      & \begin{tabular}[c]{@{}c@{}}1\\ (0)\end{tabular}            & \begin{tabular}[c]{@{}c@{}}1\\ (0)\end{tabular}& \begin{tabular}[c]{@{}c@{}}1\\ (0)\end{tabular}           & \begin{tabular}[c]{@{}c@{}}1\\ (0)\end{tabular}           & \begin{tabular}[c]{@{}c@{}}1\\ (0)\end{tabular}           \\
Time (s) & 0.26             &   4.23                                      & 0.05                                                      & 0.29                                                      & 0.1                                                       \\ \hline
\end{tabular}
\end{table}

We again observe from the results presented in Table~\ref{linear_regression_results} that the performance measures are comparable for HS-MCMC and HS-LLA. However, HS-MCMC is drastically slower and also requires a choice for the credible interval width for variable selection; which is not the case under HS-LLA. This underlines the scalability of our approach for Bayesian sparse signal recovery.
\section{Concluding Remarks}
We provide a new Laplace mixture representation of the popular horseshoe prior. While this can be used as an alternative generative hierarchy for the horseshoe, similar to the usual representation of half Cauchy scale mixture of normals, our results show horseshoe prior enjoys some particularly nice properties due to the resultant complete monotonicity of the density function. This is relevant because not all normal scale mixtures are also automatically Laplace mixtures. For example, the class of power exponential densities, $p_{X}(x) \propto \exp(-|x|^\alpha)$, can be written as a normal scale mixture with respect to a positive $\alpha/2$ stable variable for $\alpha \in (1,2)$ \cite{west1987scale}, but these densities are not Laplace mixtures. In addition, our results imply a numerically stable LLA approach for penalized likelihood estimation problem under the horseshoe, and establishes LLA as  special case of EM, thus proving LLA in this case results in the MAP estimate, which is sparse, due to the soft thresholding step in LLA. Deeper studies of these connections should be considered future work.
\appendix
\emph{The hyperparameter $\tau$.}  Proposition~\ref{prop_equiv_HS_Laplace} is presented for a fixed $\tau>0$. In both fully Bayesian and penalized likelihood settings, $\tau$ needs to be tuned for practical implementation. We used cross validation for our numerical results in Section~\ref{sec:num}. Fully Bayesian approaches require a hyperprior on $\tau$, popular among them the half-Cauchy~\cite{polson2012half}. Other techniques for choosing $\tau$ include Akaike information criterion or AIC~\cite{lingjaerde2022scalable}, marginal likelihood~\cite{bhadra2022graphical} and effective model size~\cite{piironen2017sparsity}. For a detailed survey, see \cite[Section 5]{bhadra2019lasso}.

\emph{Fast computation of $\lambda_{(k)}$.} Computing the ratio of integrals in Equation~\eqref{derivative_peanlty}, at every iteration of the optimization procedure might be perceived as a potential drawback, for two possible reasons: (a) closed form of the integrals do not exist and (b) numerical integration might be a time intensive exercise. Luckily, near-minimax rational approximation of the Dawson function exists~\cite{lether1997constrained}, and is given by:
\begin{equation*}
    \frac{D_+(u)}{u} \approx G(u) = \frac{1+\frac{33}{232}u^2+\frac{19}{632}u^4+\frac{23}{1471}u^6}{1+\frac{517}{646}u^2+\frac{58}{173}u^4+\frac{11}{262}u^6+\frac{46}{1471}u^8}\cdot
\end{equation*}
\noindent This approximation has a maximum relative error of $6.1\times 10^{-4}$ for $|u|<\infty$, is exact at $u=0$ and has the asymptotic property, $uG(u)\sim D_+(u)$, as $u\rightarrow \infty$~\cite{lether1997constrained}. Therefore, having the values of $D_+(u)$ pre-computed and stored on a large enough and fine grid of $u$, a simple Riemann sum can be used to evaluate the expressions in Equation~\eqref{derivative_peanlty}. Lipschitz continuity of the rational approximation can also be established with some further algebra, aiding a selection of grid size.
\IEEEtriggeratref{17}
\bibliographystyle{IEEEtran}
\bibliography{ref}

\end{document}